\documentclass{article}
\usepackage{ijcai24}

\usepackage{times}
\usepackage{soul}
\usepackage{url}
\usepackage[hidelinks]{hyperref}
\usepackage[utf8]{inputenc}
\usepackage[small]{caption}
\usepackage{graphicx}
\usepackage{amsmath}
\usepackage{amsthm}
\usepackage{booktabs}
\usepackage{algorithm}
\usepackage{algorithmic}
\usepackage[switch]{lineno}

\pdfpagewidth=\paperwidth
\pdfpageheight=\paperheight

\newtheorem{theorem}{Theorem}
\newtheorem{lemma}{Lemma}
\usepackage{subfigure}
\usepackage{xcolor}
\newcommand{\rOneTMK}[1]{\textcolor{black}{#1}}

\title{MGCBS: An Optimal and Efficient Algorithm for Solving Multi-Goal Multi-Agent Path Finding Problem}

\author{
Mingkai Tang$^1$
\and
Yuanhang Li$^1$\and
Hongji Liu$^1$\and
Yingbing Chen$^1$\and
Ming Liu$^2$\And
Lujia Wang$^2$ \\
\affiliations
$^1$Hong Kong University of Science and Technology\\
$^2$Hong Kong University of Science and Technology (Guangzhou)\\
\emails
\{mtangag, yliog, hliucq, ychengz\}@connect.ust.hk,
\{eelium, eewanglj\}@hkust-gz.edu.cn
}

\begin{document}

\maketitle

\begin{abstract}
With the expansion of the scale of robotics applications, the multi-goal multi-agent pathfinding (MG-MAPF) problem began to gain widespread attention. This problem requires each agent to visit pre-assigned multiple goal points at least once without conflict. Some previous methods have been proposed to solve the MG-MAPF problem based on Decoupling the goal Vertex visiting order search and the Single-agent pathfinding (DVS). However, this paper demonstrates that the methods based on DVS cannot always obtain the optimal solution. To obtain the optimal result, we propose the Multi-Goal Conflict-Based Search (MGCBS), which is based on Decoupling the goal Safe interval visiting order search and the Single-agent pathfinding (DSS). Additionally, we present the Time-Interval-Space Forest (TIS Forest) to enhance the efficiency of MGCBS by maintaining the shortest paths from any start point at any start time step to each safe interval at the goal points. The experiment demonstrates that our method can consistently obtain optimal results and execute up to 7 times faster than the state-of-the-art method in our evaluation. 
\end{abstract}

\section{Introduction}
With the development of the robotic industry, the multi-agent system has attracted more and more attention \cite{salzman2020research,stern2019multi,tjiharjadi2022systematic}. One of the critical problems to be solved is the multi-agent path finding (MAPF) problem. The MAPF problem requires planning a conflict-free path for each agent from its starting point to its goal point.
MAPF is involved in many practical application scenarios in the real world, such as aircraft towing vehicles \cite{morris2016planning}, video games \cite{ma2017feasibility} and traffic management \cite{choudhury2022coordinated,dresner2008multiagent}. 

In the MAPF problem, each agent can be assigned only one goal point. This setting does not meet the needs of some large-scale robot applications.
For example, in an automated warehouse scenario, each robot may need to deliver multiple goods in one trip. In this case, the robot needs to be provided with a collision-free path with multiple goal points.
This problem can be modeled as a multi-goal multi-agent pathfinding (MG-MAPF) problem.
The solver of MG-MAPF needs to calculate a collision-free path for each agent so that the agent can visit each of its goals at least once in an arbitrary visiting order. Figure \ref{fig:instance} shows an example of MG-MAPF.

\begin{figure}[t]
\centering
\includegraphics[width=0.24\textwidth]{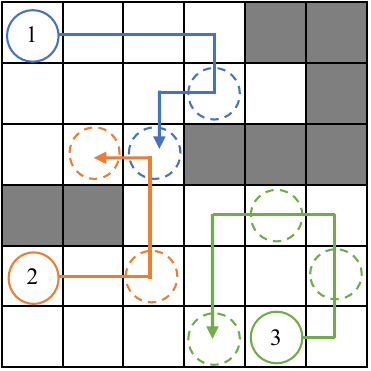} 
\caption{An example of the MG-MAPF with three agents. The grey cells represent the impassable areas occupied by obstacles.
The solid circles indicate the start point of the agent, and the number marked inside is the agent ID. The dotted circles represent the goal points of the agent in the corresponding color. The colored arrow marks the path that can visit all goal points from the start point of the agent.}
\label{fig:instance}
\end{figure}

Solving the MG-MAPF problem optimally, even in degenerate scenarios, can be time-consuming. The open loop traveling salesman problem \cite{applegate2011traveling}, which is widely recognized as an NP-hard problem, can be reduced to a subset of the MG-MAPF that only considers one agent. Furthermore, the classical MAPF problem, which is proved to be NP-hard \cite{yu2013structure}, can also be reduced to a subset of the MG-MAPF, which considers only one goal for each agent. Therefore, it can be concluded that the problem of optimally solving the MG-MAPF is NP-hard.


\rOneTMK{
Some methods have been proposed to solve the MG-MAPF problem by Decoupling the goal Vertex visiting order search and the Single-agent pathfinding (DVS), such as the Hamiltonian Conflict-based Search \cite{surynek2021multi}, which is the current state-of-the-art (SOTA) method for the MG-MAPF problem. We named the MG-MAPF solver that is based on DVS as the DVS method.
To find the shortest path to visit all goals for a single agent under constraints, DVS methods search a goal vertex visiting order iteratively. At each iteration, an unvisited goal is enumerated and tried to append to the end of the order. The corresponding path is constructed by concatenating the shortest paths under constraints between two neighbor goals on the goal vertex visiting order.}

\rOneTMK{
However, this paper provides evidence through the case study and experiment that DVS methods cannot always obtain the optimal result for the MG-MAPF problem.
We introduce a new approach, the Multi-Goal Conflict-Based Search (MGCBS), to solve the MG-MAPF problem optimally and efficiently. In contrast to the DVS method, MGCBS is based on Decoupling the goal Safe interval visit order search and the Single-agent pathfinding (DSS). }
The safe interval (SI) refers to a safe configuration with a maximal time period, where the `maximal' means that if this time period were to be extended by a time step in any direction, the collision would occur \cite{phillips2011sipp}. The goal safe interval (GSI) refers to the safe interval at the goal vertex. Formally, let $([t^{SI}_0, t^{SI}_1], s^{SI})$ be the SI at vertex $s$ for the maximal time interval $[t^{SI}_0, t^{SI}_1]$. $([t^{SI}_0, t^{SI}_1], s^{SI})$ is an GSI iff $s^{SI}$ is a goal vertex.
We name the MG-MAPF solver that is based on DSS as the DSS method. The DSS method searches a GSI visiting order so that at least one GSI is visited at each goal. The corresponding path is obtained by concatenating the shortest path to visit each GSI following the order. 
In addition, we propose a data structure, the Time-Interval-Space Forest (TIS Forest), to reduce redundant calculations of multiple queries in the low-level solver of MGCBS by maintaining the shortest paths and their length from any start vertex at any start time to each GSI. 
Overall, the main contributions of this paper are as follows.
\begin{enumerate}
    \item We present that the DVS methods cannot always obtain the optimal solution by case study and experiment.
    \item We introduce a two-level approach, MGCBS, to solve MG-MAPF, achieving high computational efficiency in obtaining optimal solutions.  
    \item We present the TIS Forest, which maintains the shortest paths to each GSI from any start vertex at any start time, minimizing redundant calculations of multiple queries.
    \item We provide the theoretical proof of the optimality and completeness of MGCBS.
    \item We conducted a comprehensive experimental evaluation and compared our proposed method with the current SOTA method. Compared with the SOTA method, our method can consistently obtain the optimal solution while achieving a maximum speedup ratio of up to 7 in our evaluation.
\end{enumerate}

\section{Related Work}

Some variants of the MAPF problem, which consider more than one goal point for an agent, have recently been studied. The Multi-Agent Pickup-and-Delivery (MAPD) problem and its variants, which require the agent to pick up the object in one location and deliver it to another location, were studied in \cite{vcap2015prioritized,ma2017lifelong,xu2022multi}.
In \cite{zhang2022multi}, the Multi-Agent Path Finding with Precedence Constraints (MAPF-PC) problem was proposed, where the visiting order of goal points needs to satisfy some precedence constraints. The Multi-Agent Simultaneous Multi-Goal Sequencing and Path Finding (MSMP) problem, which requires assigning goals to each agent before pathfinding, was solved in \cite{ren2021ms,ren2022conflict}. However, the settings of the above problems differ from the MG-MAPF problem, making their methods not directly usable in the MG-MAPF problem.

The MG-MAPF problem was firstly discussed in \cite{surynek2021multi}, and two solutions were proposed: Hamiltonian Conflict-based Search (HCBS) and Satisfiability Modulo Theories Conflict-based Search (SMT-HCBS). The HCBS, which can be categorized as a DVS method, typically runs faster than SMT-HCBS due to the leverage of its heuristic function. However, it has room for improvement in optimality and efficiency. 
\rOneTMK{For optimality, our method improves upon HCBS by using DSS. For efficiency, our method uses the TIS Forest to reduce redundant calculations of the multiple queries of the shortest path for each agent.}

\section{Problem Definition}
The MG-MAPF problem is defined as follows. A set of agents $A = \{a_1, a_2, ..., a_k\}$ can move on an undirected graph $G(V, E)$ where each edge is of unit length. Let $n_i$ be the number of goals of $a_i$. The task of agent $a_i$ can be described as $\Omega_a = \{s_{i}, g_{i}\}$ where $s_{i}$ is the start vertex and $g_{i} = \{g_{i}^{1}, g_{i}^{2}, ..., g_{i}^{n_i}\}$ denotes the goal vertices of the agent. At each time step, agents can choose to wait at the current vertex or move to an adjacent vertex. The agent will stay at one of its goal vertices after completing all movements without incurring any additional cost. The solution to the MG-MAPF problem is a collection of agents' collision-free paths, where the agent can start from its start vertex and visit all its goal vertices at least once with arbitrary visiting order. A collision occurs when two agents are located in the same vertex (vertex conflict) or moving along the same edge (edge conflict) at the same time step. The cost of an agent's path is the total time steps used to visit all goals. We use the summation of the cost (SOC) as the objective of the problem, meaning that the cost of the solution is the summation of the cost of each agent's path.

\section{Case Study}
This section will provide an example of the MG-MAPF problem to illustrate that the DVS method is not optimal.

\begin{theorem}
The optimality of the methods based on decoupling the goal vertex visiting order search and single-agent pathfinding cannot be guaranteed in the MG-MAPF problem.
\end{theorem}
\begin{proof}
An example that the DVS method cannot find the optimal solution is shown in Figure \ref{fig:prove}. 
In this example, agent $2$ only has one possible goal vertex visiting order, while agent $1$ has two possible situations.
If agent $1$ chooses to go to $g^1_1$ before going to $g^2_1$, the lower bound of the total cost is $207 + 7 = 214$.
If agent $1$ chooses to travel to $g^2_1$ before $g^1_1$, the lower bound of the total cost will be $111 + 7 = 118$.  
Now we consider to adopt the second choice.
In this case, an edge conflict between vertex $D$ to vertex $E$ occurs at the time step $6$, which is in the path of agent $1$ from $g^2_1$ to $g^1_1$.
Based on the definition of DVS, the path from $g^2_1$ to $g^1_1$ may be adjusted to avoid the conflict. For example, in HCBS, constraints are created on the path of $g^2_1$ to $g^1_1$ of agent $1$.
However, the path of agent $1$ from the start vertex $s_1$ to the goal vertex $g^2_1$ will remain unchanged because the time step of conflict surpasses the time step of reaching $g^2_1$.
One of the minimal-cost solutions generated by the DVS method is that agent $1$ goes to $g^2_1$ before $g^1_1$, and agent $2$ remains at vertex $C$ for $8$ time steps to give paths to agent $1$ before proceeding directly to $g^1_2$. The total cost is $111 + 15 = 126$.
However, there exists an alternative solution with a lower total cost that deviates from the shortest path for agent $1$ starting from $s_1$ towards $g^2_1$. 
Agent $1$ remains at vertex $s_1$ for $2$ time steps. Subsequently, it proceeds to vertex $g^2_1$ and then travels to vertex $g^1_1$. 
In this solution, the total cost is $113 + 7 = 120$, which is better than the minimal cost solution given by the DVS method among all goal vertex visiting orders. 
\end{proof}
\rOneTMK{
The reason why the DVS method is not optimal is because it incorrectly assumes that the path reaching the goal vertex earlier is always better than the path reaching it later. Under the given goal vertex visiting order, it only generates one path to reach each goal at the earliest possible time step. However, in some cases, the optimal path might not reach a subset of goal vertices at the earliest possible time step. The DVS method misses these paths.}

\rOneTMK{
To obtain the optimal path, we propose to search based on DSS, which only assumes that the path reaching the GSI earlier is always better than the path reaching it later. In contrast to the DVS method, DSS can obtain multiple paths under a given goal vertex visiting order when there are multiple GSIs at some goal vertices. Considering the same case above, when the path between $g^2_1$ to $g^1_1$ is adjusted and causes conflicts at $g^2_1$ at a time step larger than the earliest time step to reach $g^2_1$, more than one GSI will appear in $g^2_1$, providing the potential to search the non-earliest path to $g^2_1$ from $s_1$.}

It is observed that the DVS method is more likely to obtain a non-optimal result in crowded scenarios. On the one hand, in scenarios with few agents and obstacles, the agent can temporarily move to the neighbor vertex to give paths to another agent and return to the original vertex once another agent has passed. In this way, the DVS method can generate a path from an earlier SI to a later SI without missing any potential optimal paths. On the other hand, in crowded scenarios, the agent may not be able to return to the original vertex quickly. It makes the agent infeasible for some time steps in the later SI, making the DVS method miss some potential optimal paths, which can be obtained by the DSS method.


\begin{figure}[t]
\centering
\includegraphics[width=0.43\textwidth]{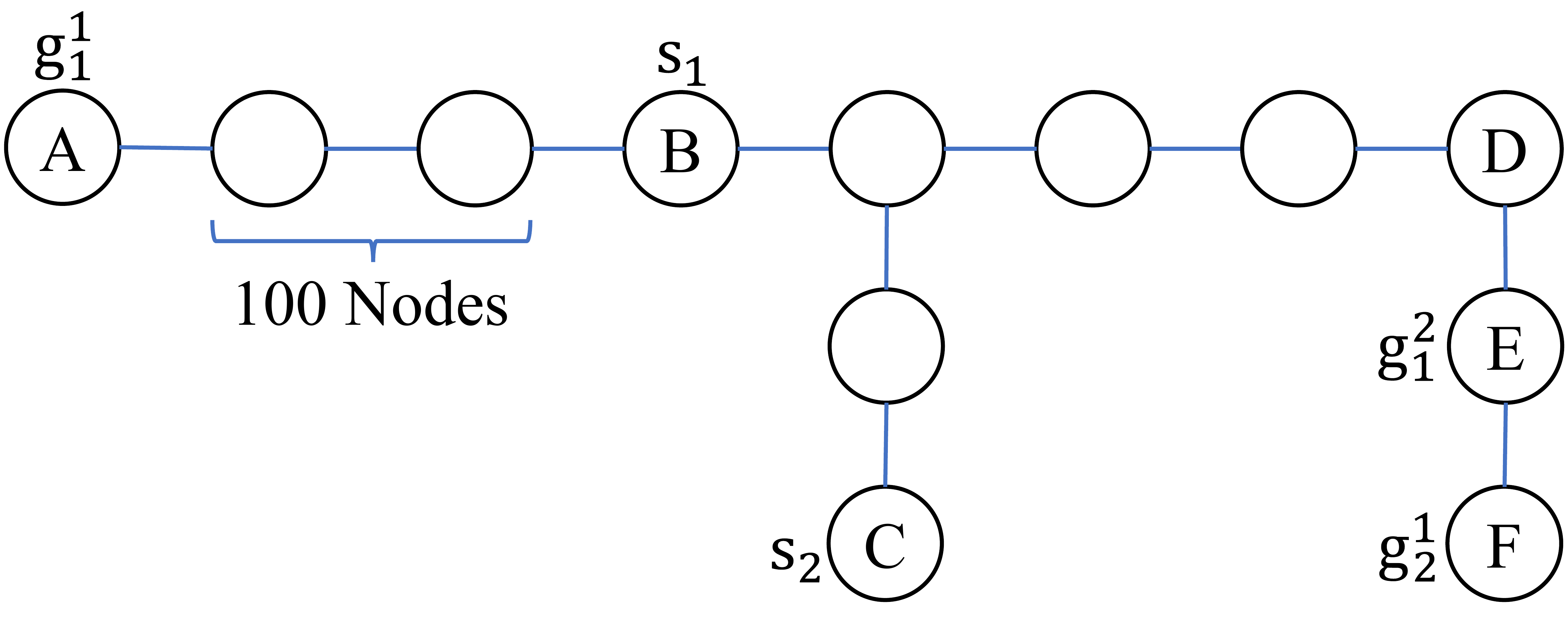} 
\caption{An example of the MG-MAPF with two agents is represented in a planar graph. The agent $1$ starts at vertex $B$ and has two goals located at vertex $A$ and vertex $E$. The agent $2$ starts at the vertex $C$ and only has one goal at vertex $F$. }
\label{fig:prove}
\end{figure}

\section{Methodology}

We propose an optimal and efficient two-level method, the MGCBS, for the MG-MAPF problem based on DSS. To improve the efficiency of the search process, we introduce the TIS Forest data structure to reduce redundant calculations of multiple queries. A TIS Forest corresponds to an agent and is composed of several Time-Interval-Space Trees (TIS Tree), which corresponds to a specific GSI of the agent. The TIS Tree can be used to get the shortest path and its length to the corresponding GSI from any start vertex at any start time step.

\subsection{MGCBS}
The MGCBS is a two-level solver that can effectively calculate the optimal path for each agent to visit all goals without conflict. The high-level solver employs a constraint tree to manage conflicts between agents, while the low-level solver uses an A*-based solver to find the best GSI visiting order.
\subsubsection{High-level Solver}
The high-level solver of the MGCBS is an extension of the high-level solver of the conflict-based search (CBS) \cite{sharon2015conflict}. It builds a constraint tree to solve conflicts between different agents.
The constraints consist of vertex constraints, which prevent an agent from occupying a vertex at a specific time step, and edge constraints, which prohibit an agent from traversing an edge at a given time step. Each constraint tree node saves the TIS Forests, the constraint set, each agent's path, and the value of SOC. The implementation of the TIS Forest will be discussed in subsection \ref{subsection:sif}.


In the high-level solver, a distance table $D$ is built to store the distance from each vertex to each goal vertex for each agent, which will be used in the low-level solver. It is observed that there are no vertex constraints at the root node of the constraint tree, causing only one GSI with a whole time interval at each goal vertex. Therefore, $D$ can be built by querying the TIS Forest at the root node of the constraint tree.

The differences between the high-level solver of basic CBS and MGCBS are the operations on the TIS Forest and the distance table.
At the beginning of the high-level solver of MGCBS, the TIS Forests in the root node are built and used to construct the distance table. 
Then, the TIS Forests, the distance table, and the agents' task information are fed to the low-level solver to calculate the path that visits all goal vertices without considering other agents. 
When a new constraint tree node is generated, its TIS Forests are copied from its parent node, and one of them is reconstructed by the new constraint set. 
The new TIS Forest, the distance table, and the agent's task information are put into the low-level solver to compute the shortest path under the constraints.

Algorithm \ref{alg:high} shows the pseudocode of the high-level solver. In lines 1  $\sim$ 9, the root node is constructed and put into the open set. The TIS Forests are built in line 4, and the distance table is constructed in line 5. The initial path of each agent is calculated in line 6. In lines 11 $\sim$ 16, the minimum cost node is found and checked whether it contains a conflict. If no, the final solution is found. Otherwise, the constraints are built according to the earliest conflict. 
In lines 19  $\sim$ 27, new nodes are constructed for each constraint, and the corresponding constraint is added to the constraint set. The TIS Trees that are related to the constrained agent are reconstructed according to the new constraint set in line 22. The low-level solver is called to build the path of a single agent in line 23.

\begin{algorithm}[tb]
\small
\caption{$MGCBS_{high}$}
\label{alg:high}
\textbf{Input}: agents $A$, graph $G$

\begin{algorithmic}[1] 
\STATE  $R$ $\leftarrow$ new node
\STATE  $R.cons$ $\leftarrow$ $\emptyset$
\FOR{ each agent $a$ in $A$ }
\STATE  $R.tisf[a] \leftarrow BuildTISForest(G, a, \emptyset)$
\STATE $D[a_i] \leftarrow BuildDistanceTable(R.tisf[a])$
\STATE $R.paths[a] \leftarrow$ \\ \quad $MGCBS_{low}(R.tisf[a], D[a], a, \emptyset)$ // Algorithm 2
\ENDFOR
\STATE $R.cost$ $\leftarrow$ calculate the SOC of $R.paths$.
\STATE $OPEN$ $\leftarrow$ \{$R$\}
\WHILE{$OPEN$ $\neq$ $\emptyset$}
\STATE $N$ $\leftarrow$ minimum cost node from $OPEN$.
\STATE $OPEN$ $\leftarrow$ $OPEN \backslash \{N\}$
\STATE $F$ $\leftarrow$ the earliest collision in $N$
\IF{$F$ is $None$}
\RETURN $N.paths$
\ENDIF
\STATE $C$ $\leftarrow$ build constraints from $F$
\FOR{constraint $c=(a, t, v)/(a, t, v_i, v_j)$ in $C$}
\STATE $P$ $\leftarrow$ new node
\STATE $P.tisf$, $P.paths$ $\leftarrow$ $N.tisf$, $N.paths$
\STATE $P.cons$ $\leftarrow$ $N.cons$ $\cup$ $c$
\STATE $P.tisf[a] = BuildTISForest(G, a, P.cons)$
\STATE $P.paths[a]$ $\leftarrow$  $MGCBS_{low}$($P.tisf[a],$\\ \quad \quad \quad \quad \quad \quad \quad \quad \quad $D[a], a, P.cons$) // Algorithm 2
\IF{$P.paths[a]$ is not $NULL$}
\STATE $P.cost$ $\leftarrow$ calculate the SOC of $P.paths$
\STATE $OPEN$ $\leftarrow$ $OPEN$ $\cup$ \{$P$\}
\ENDIF
\ENDFOR
\ENDWHILE
\RETURN $NULL$
\end{algorithmic}
\end{algorithm}

\begin{algorithm}[tb]
\small

\caption{$MGCBS_{low}$}
\label{alg:low}
\textbf{Input}: TIS Forest $tisf$, distance table $D$, agent $a$, constraint set $cons$

\begin{algorithmic}[1] 
\STATE  $R$ $\leftarrow$ new node
\STATE  $R.L$ $\leftarrow$ $\emptyset$
\STATE  $R.p$ $\leftarrow$ $GetEarliestSI(a.start, cons)$
\STATE $R.g \leftarrow 0$
\STATE $R.h \leftarrow Cost(GetMST(D, R.L, (R.p).v))$
\STATE $R.f \leftarrow R.g + R.h$
\STATE $OPEN, CLOSED$ $\leftarrow$ \{$R$\}, $\emptyset$
\WHILE{$OPEN$ $\neq$ $\emptyset$}
\STATE $N$ $\leftarrow$ minimum $f$ node from $OPEN$.
\STATE $OPEN$ $\leftarrow$ $OPEN \backslash \{N\}$
\STATE $CLOSED$ $\leftarrow$ $CLOSED \cup \{N\}$
\IF{$N.L$ contains all goals of $a$}
\STATE $order \leftarrow GenerateGoalVisitedOrder(N)$
\STATE $path \leftarrow GeneratePath(tisf, order)$
\RETURN $path$
\ENDIF
\FOR{$q$ in $GetAllPotentialGSI(a, cons, N.L)$}
\STATE $M$ $\leftarrow$ $GetNode$($OPEN$, $CLOSED$, $N.L\cup \{q.v\}$, $q$)
\IF{$M \notin CLOSED$}
\STATE $tist \leftarrow$ $GetTISTree$($q$, $tisf$)
\STATE $len \leftarrow GetPathLength(tist, (N.p).v, N.g)$
\IF{$M \in OPEN$}
\STATE $M.g \leftarrow min(M.g, N.g+len)$
\STATE $M.f \leftarrow M.g + M.h$
\ELSE
\STATE $M.g \leftarrow N.g+len$
\STATE $M.h \leftarrow Cost(GetMST(D, M.L, (M.p).v))$
\STATE $M.f \leftarrow M.g + M.h$
\STATE $OPEN$ $\leftarrow$ $OPEN$ $\cup$ \{$M$\}
\ENDIF
\ENDIF
\ENDFOR
\ENDWHILE
\RETURN $NULL$
\end{algorithmic}
\end{algorithm}

\begin{figure}[t]
\centering
\includegraphics[width=0.45\textwidth]{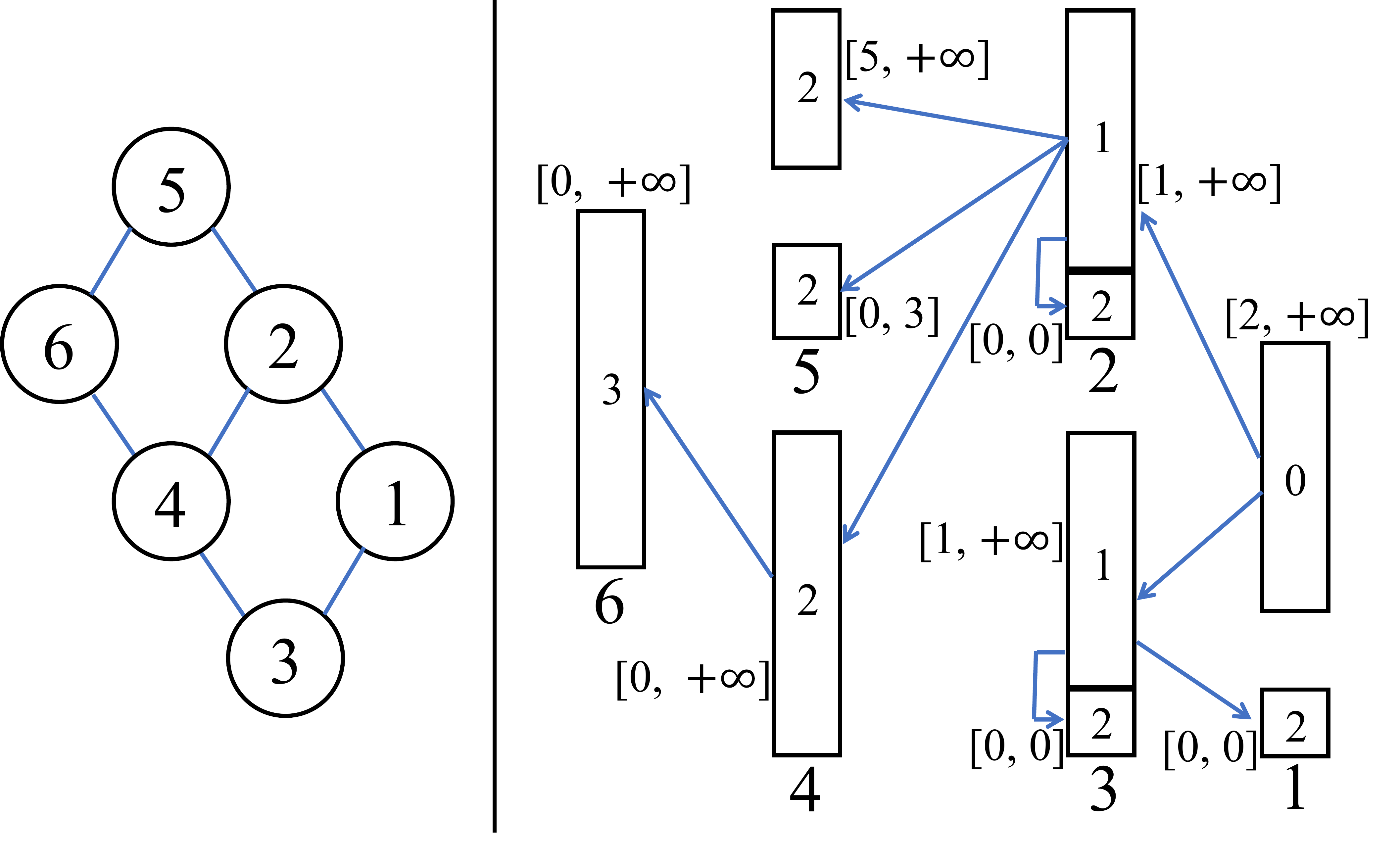} 
\caption{
The left subfigure shows the graph on which the TIS Tree was built. The number in the circle is the vertex ID. The right subfigure shows the TIS Tree whose seed GSI is $[2, +\infty]$ at vertex $1$. The constraint set includes two vertex constraints, including vertex $1$ at time step $1$, and vertex $5$ at time step $4$. The block represents the node in the tree, where the number in the block is the cost of the node. The number under the block is the vertex ID. The bracketed number represents the time interval of the TIS state of the node. The arrow represents the edge of the TIS Tree pointing from the parent node to the child node. The self-point arrow at vertices $2$ and $3$ indicates waiting for one time step. \label{fig:sitree} }

\end{figure}

\subsubsection{Low-level Solver}
The low-level solver computes the shortest path for a single agent under constraints, ensuring that each goal is visited at least once based on DSS. The low-level solver comprises two stages: the GSI visiting order search stage and the path-generating stage.

The GSI visiting order stage search uses an A* solver for the best GSI visiting order. The state in the search can be represented by $(L, p)$, where $L$ is the visited goal set and $p$ is the SI where the agent is currently located.
The cost of the state $(L, p)$ is the length of the path that has visited all the goals in $L$ and is currently located in $p$. Let the $g$ be the cost value which is the same as the current time step, the $h$ be the heuristic value, and the $f$ be the evaluation value. The cost of the minimum spanning tree (MST) of the currently located vertex and all unvisited goal vertices is used as the $h$. The distance table $D$ can be utilized to construct the MST. During the search, the state $(L, p)$ can transfer to $(L \cup \{ q.v \}, q)$, where the $q$ is a GSI whose vertex is not in $L$, and the $q.v$ refers to the vertex corresponding to the $q$. It means that the agent moves from the $p.v$ at the start time step $g$ to the GSI $q$, whose vertex is an unvisited goal vertex $q.v$. 
The transfer cost is the shortest path length from the $p.v$ at time step $g$ to the GIS $q$ under the constraint set. It can be directly obtained from the TIS Forest. Specially, if $q.v$ is the final unvisited goal, only the latest GSI at $q.v$ can be chosen as the next GSI for avoiding conflicts after the agent finishes at the final goal.

The path-generating stage is executed when the minimum-cost state that visits all goals is found. We backtrack the final state to the initial state to get the GSI visiting order and use the TIS Forest to build the path based on the order.

Algorithm \ref{alg:low} shows the pseudocode of the low-level solver. In lines 1 $\sim$ 7, the initial node is initialized and put into the open set, while the $R.p$ is the earliest SI at the start vertex of the agent. The heuristic value of the node is calculated by the cost of the MST in line 5. At each iteration, the node with the minimum $f$ value is popped from the open set (lines 9 $\sim$ 11). If the current node $N$ has already visited all goals, extract the GSI visiting order and then build the path (lines 12 $\sim$ 16). Otherwise, all GSI at unvisited goals is enumerated with the exception of the last unvisited goal, for which only the lastest GSI is considered (lines 17). Let $q$ be the current enumerated GSI. A node $M$ is constructed where $M.L = N.L \cup \{q.v\}$ and $M.p = q$ (line 18). The TIS Tree corresponding to $q$ is filtered from the TIS Forest and used to obtain the path length from $(N.p).v$ at the time step $N.g$ to $q$ (lines 20 $\sim$ 21). $M$ is updated if it can be improved by $N$ (lines 22 $\sim$ 30).

\subsection{Time-Interval-Space Forest \label{subsection:sif}}
The shortest path length to the GSI is frequently queried in the low-level solver. If the MGCBS employs the A* algorithm to obtain the path length, it will become time-consuming. Observing that the shortest path length to a specific GSI may be queried multiple times during the search, we consider utilizing a data structure to reduce the redundant computation. It is not trivial because the start time of each query is unknown before the search and can only be obtained after the path to the previous GSI is generated. Considering that the constraints are related to the time step, the optimal path of different start time steps might be diverse, making it hard to reuse the result of the previous search. 

We propose TIS Forest to reduce redundant calculations. Each TIS Forest corresponds to an agent. It consists of several TIS Trees, each corresponding to a GSI of the agent. 
The TIS Tree maintains the shortest path and its length from any start vertex at any start time to the GSI. We name that GSI as the seed GSI of the TIS Tree. 

Let $(t,s)$ be a time-space state (TS state) representing that the agent is located at the vertex $s$ at the time step $t$. Let $([t_0, t_1], s)$ be a time-interval-space state (TIS state), representing a collection of TS states at the same vertex, i.e., $([t_0, t_1], s) = \{(t,s)| t\in[t_0, t_1]\}$. A TIS state is safe if and only if it does not contain any TS states under vertex constraints. It should be noted that the SI is a special type of TIS state, while the TIS state does not always need to be maximal.

Each node in the TIS Tree represents a safe TIS state, containing the TS states that take the same vertex sequence as the shortest path to the seed GSI. Therefore, different TS states in a TIS state have the same shortest path length to the seed GSI. We define the cost of the node as the shortest path length of the TS states in the TIS state to the seed GSI. There might be more than one node at a vertex when some vertex constraints exist at the vertex. 

The TIS Tree is constructed by a Dijkstra-based algorithm \cite{dijkstra1959note}. Initially, we find all maximal TIS states (i.e., SI) at each vertex and create a node for each of them. We refer to the node corresponding to the seed GSI as the seed node. If the node is the seed node, its cost is set to $0$; otherwise, its cost is set to $+\infty$. 
At the beginning, all nodes are put into an unvisited set $U$. 
At each iteration of the search, the node with the minimum cost in $U$ pops out and is used to improve the path of its neighbor through the reverse edge. 
Specifically, let $([t_0, t_1], v)$ be the TIS state of the current minimum cost node $n$ and $cost(n)$ be the cost of $n$. 
Let $v'$ be a vertex that can take one action to transfer to $v$, i.e., $v'\in Nbr(v)$ and $Nbr(v) = \{v' | (v', v) \in E\} \cup \{v\}$. 
Now we consider how to improve the path of the nodes at $v'$ by $([t_0, t_1], v)$. We construct a TIS state set denoted by $B$ at vertex $v'$, such that each TS state in the TIS state is safe and can transfer to a TS state in $([t_0, t_1], v)$ at a time step. For example, $([t_0, t_1], v)=([2, 8], v)$ and here is a vertex constraint at vertex $v'$ at time step 3 and an edge constraint from vertex $v'$ to vertex $v$ at time step 6, the $B$ will be $\{([1,2],v'), ([4,5], v'), ([7,7], v')\}$.
We enumerate all nodes at $v'$ and all TIS states in $B$. Let $([t_0', t_1'], v')$ be the TIS state of the current enumerated node $n'$ and $([t_0^B, t_1^B], v')$ be the TIS state in $B$. 
If $cost(n') > cost(n) + 1$ and $[t_0', t_1']$ can be fully covered by $[t_0^B, t_1^B]$, the $n'$ can be improved by setting $n$ as its parent and $cost(n') =  cost(n) + 1$. If $cost(n') > cost(n) + 1$ and only a part of $[t_0', t_1']$ is covered by $[t_0^B, t_1^B]$, $n'$ will be divided into several new nodes according to the coverage, and only the new node whose time interval is fully covered by $[t_0^B, t_1^B]$ can be improved by $n$. 
When a node needs to be divided, the $U$ is updated by deleting the original node and adding the new nodes.
The search will stop when the $U$ is empty. Figure \ref{fig:sitree} shows an example of the TIS Tree. 

After the construction, the TIS Tree can be used to query the shortest path from any start vertex at any start time to the seed GSI by the following steps. Firstly, we construct a TS state according to the start time and vertex. Secondly, we find out the node whose TIS state contains the TS state. 
Finally, we backtrack this node to the seed node and construct the shortest path through the vertices of passed nodes. Furthermore, we can get the shortest path length directly using the cost of the node found in the second step without building the path explicitly.
When we need to obtain the shortest path or its length to a GSI by a TIS Forest, we can extract the TIS Tree that corresponds to the GSI and use it for the query.

The concept of TIS Forest may seem similar to the SIPP algorithm \cite{phillips2011sipp} as both the TIS Forest and the SIPP algorithm search based on time intervals. However, their underlying principles differ. The SIPP algorithm utilizes a forward search, with each SI containing a single dominant time step, rendering all other time steps unimportant for the search. This property allows the SIPP algorithm to reduce the total number of nodes during the search. In contrast, the TIS Forest employs a backward search. On the one hand, we cannot use a mechanism similar to the SIPP algorithm to construct the TIS Forest because the precise goal-reaching time step is unknown beforehand, and the dominant time step doesn't exist. On the other hand, the TIS Forest is built by the principle that the shortest path to the seed GSI from all TS states within a TIS state have the same vertex sequence, and all TS states within a TIS state can be expanded simultaneously along a reverse edge.


\section{Theoretical Analysis \label{ss:analysis}}
In this section, we will prove the optimality and completeness of MGCBS.
\begin{lemma}
For two TS states in the same SI, the minimum completion time of the path to visit a set of goals at least once from the earlier TS state will not exceed the time from the later TS state.
\end{lemma}
\begin{proof}
We prove it by contradiction.
Let ($t_0$, s) and ($t_1$, s) be two TS states in the same SI where $t_1 > t_0$.
Define $C(t,s)$ as the minimum completion time of the path to visit a set of goals at least once from a TS state $(t,s)$.
Suppose, towards a contradiction, that $C(t_1, s) < C(t_0, s)$.
Consider the optimal path that achieves the completion time $C(t_0, s)$.
This path can be modified to stay at vertex $s$ until time $t_1$ and then follow the same sequence of vertices as the optimal path starting from $(t_1, s)$. The modified path would visit all goals no later than the optimal path from $(t_1, s)$, yielding a completion time that is at most $C(t_1, s)$, in direct contradiction to the supposition that $C(t_1, s) < C(t_0, s)$.
\end{proof}
\begin{lemma}
The path maintained in the TIS Tree is the optimal path to the seed GSI of the TIS tree under given constraints.
\end{lemma}
\begin{proof}
  The TIS Tree is constructed based on a backward version of the Dijkstra algorithm \cite{dijkstra1959note}. The optimality of the TIS Tree can be guaranteed by the optimality of the Dijkstra algorithm.
\end{proof}

\begin{lemma}
The function $GeneratePath$ in Algorithm \ref{alg:low} can obtain the optimal path following a given GSI visiting order under given constraints.
\end{lemma}
\begin{proof}
The resulting path of $GeneratePath$ is constructed by iteratively concatenating the subpath to the next GSI by the TIS Tree, and the length of the subpath is shortest according to the Lemma 2.
The current Lemma can be proved by induction. The search starts from a TS state with time step 0. Assume that using this construction method can obtain the path with the minimum completion time following the first $k$ GSI visiting order. According to Lemma 1, concatenating the shortest subpath to the $k+1$ GSI can obtain the path with the minimum completion time following the first $k+1$ GSI visiting order. Therefore, by the principle of induction, the final resulting path has the minimum completion time following the whole GSI visiting order. 
\end{proof}

\begin{theorem}
MGCBS is an optimal solver of the MG-MAPF problem.
\end{theorem}
\begin{proof}
In the low-level solver, the heuristic function is admissible and satisfies the consistency assumption according to the property of MST. The optimality of the low-level solver can be guaranteed by the Lemma 3 and the optimality of A* \cite{hart1968formal}.
According to the optimality of CBS and the low-level solver, the optimality of MGCBS can be guaranteed.
\end{proof}

\begin{theorem}
MGCBS is a complete solver of the MG-MAPF problem.
\end{theorem}
\begin{proof}
The completeness can be proven by following steps.
Firstly, based on the TIS forest design, if there is a feasible path $r$ from one start/goal vertex $u$ at time step $t_u$ to reach one goal vertex $v$ at time step $t_v$, there will not be a vertex constraint on vertex $v$ at time step $t_v$. A TIS Tree must exist, whose seed GSI includes time step $t_v$ at vertex $v$. This node can be iteratively expanded following the reverse direction of $r$ to reach the node whose TIS state includes time step $t_u$ at vertex $u$. Therefore, the TIS Tree is complete. Secondly, if there is a feasible solution for a multi-goal single-agent pathfinding problem, their corresponding GSI visiting order could be searched in the low-level solver, and the path could be constructed according to the completeness of A* and TIS Tree. Thirdly, MGCBS is complete according to the completeness of CBS and the low-level solver.
\end{proof}

\section{Experiment}
We verify our proposed algorithm's effectiveness and optimality on the 4-neighbor grid maps. 
We randomly sample each agent's start and goal points in the grid map, ensuring that the start points of the different agents are distinct. 
The test computer is equipped with an I9-7900X CPU with 3.3 GHz and 32GB RAM. The code is publicly available at
https://github.com/tangmingkai/MGCBS.

We use the following four algorithms for the experiments.
\begin{itemize}
\item \textbf{HCBS (A1)}: A three-level algorithm \cite{surynek2021multi}. The highest level solver is a CBS algorithm, and the middle level solver is an A* algorithm that searches a goal vertex visiting order. The lowest level solver is a classic A* for single-agent pathfinding. It is the current SOTA method based on DVS. 
\item \textbf{MGCBS without TIS Forest (A2)}: MGCBS with a modified low-level solver that uses the A* algorithm to obtain the shortest path and its length to a GSI. 
\item \textbf{MGCBS (A3)}: Our propose method. 
\item \textbf{CBS + A* (A4)}: An optimal algorithm that is the CBS with a coupled low-level search using an A* algorithm to search for the shortest path visiting all the goals.
\end{itemize}
\subsection{Experiment for Efficiency}
From the MAPF benchmark \cite{sturtevant2012benchmarks}, three grid maps from small to large are selected, namely `maze-32-32-4', `lak303d', and `orz900d'. They are marked as M1, M2, and M3, as shown in Figure \ref{fig:map}.
We generate 100 instances for each number of agents in the range of $\{2,4,6,8\}$. We fix each agent's goal number to $12$. We employ the average running time and the success rate as evaluation metrics for comparing the performance of A1, A2, and A3. A test case is unsuccessful if its running time exceeds 60 seconds, in which case the running time will be directly set to 60 seconds.
\begin{table}[tb]

    \small
    \centering
    \begin{tabular}{ccrrrr}
        \toprule
         Map & $k$ & A1&A2&A3 \\  
        \midrule
M1&2&98\% & 98\% &\textbf{100\%} \\ 
&4&69\% & 69\% &\textbf{86\%}  \\ 
&6&22\% &22\% &\textbf{50\%}  \\ 
&8&4\% &4\% &\textbf{9\%}  \\ 

M2&2&\textbf{100\%} &\textbf{100\%} &98\%   \\ 
&4&\textbf{93\%} &\textbf{93\%} &91\%   \\ 
&6&66\% &66\% &\textbf{74\%}   \\ 
&8&20\% &20\% &\textbf{50\%}  \\

M3&2&0\% &0\% &\textbf{100\%}  \\ 
&4&0\% &0\% &\textbf{73\%}   \\ 
&6&0\% &0\% &\textbf{8\%}   \\ 
&8&0\% &0\% &0\%  \\ 
        \bottomrule
    \end{tabular}
    \caption{Table of the success rate. $k$ is the number of agents.}
    \label{tab:large_succ}
\end{table}

\begin{table}[tb]
\small
    \centering
    \begin{tabular}{ccrrrr}
        \toprule
       map & $k$ & A1&A2&A3 \\  
        \midrule

M1&2&4.37(-)&4.52(0.97)&\textbf{0.63(6.94)} \\ 
&4&27.39(-)&27.84(0.98)&\textbf{12.16(2.25)} \\ 
&6&53.44(-)&53.69(1.00)&\textbf{37.47(1.43)} \\ 
&8&58.75(-)&58.75(1.00)&\textbf{56.41(1.04)} \\

M2&2&8.87(-)&8.87(1.00)&\textbf{3.79(2.34)} \\ 
&4&23.73(-)&23.79(1.00)&\textbf{12.23(1.94)} \\ 
&6&44.40(-)&44.40(1.00)&\textbf{24.60(1.80)} \\ 
&8&57.18(-)&57.23(1.00)&\textbf{42.26(1.35)} \\ 
M3&2&60.00(-)&60.00(1.00)&\textbf{17.73(3.38)} \\ 
&4&60.00(-)&60.00(1.00)&\textbf{42.96(1.40)} \\ 
&6&60.00(-)&60.00(1.00)&\textbf{58.73(1.02)} \\ 
        \bottomrule
    \end{tabular}
    \caption{Table of the average running time and the speedup ratio. The first number in the cell shows the average running time (sec), and the number in parentheses indicates the speedup ratio relative to A1. $k$ is the number of agents. We omit the result of $k=8$ on M3 because there are no successful instances.}
    \label{tab:large_run}
\end{table}

\begin{table}[tb]
\small
    \centering
    \begin{tabular}{cccrrrr}
        \toprule
     Map & Algorithm &  $SN$ & $DN$ & $MRE$ & $ARE$ &\\  
        \midrule
M4&A1&841&29&17.39\% &0.19\%  \\
  &A3&852&0&0.00\% &0.00\%  \\
M5&A1&896&4&7.14\% &0.03\%  \\
  &A3&896&0&0.00\% &0.00\%  \\
        \bottomrule
    \end{tabular}
    \caption{Table for evaluating the solution quality of A1 and A3. $SN$ is the number of instances both the target algorithm and A4 can successfully run. $DN$ is the number of instances where the solution cost produced by the target algorithm differs from those produced by A4. $MRE$ and $ARE$ are maximal and average relative error, respectively. }
    \label{tab:opt}
\end{table}

\begin{figure}[t]
\centering
\includegraphics[width=0.48\textwidth]{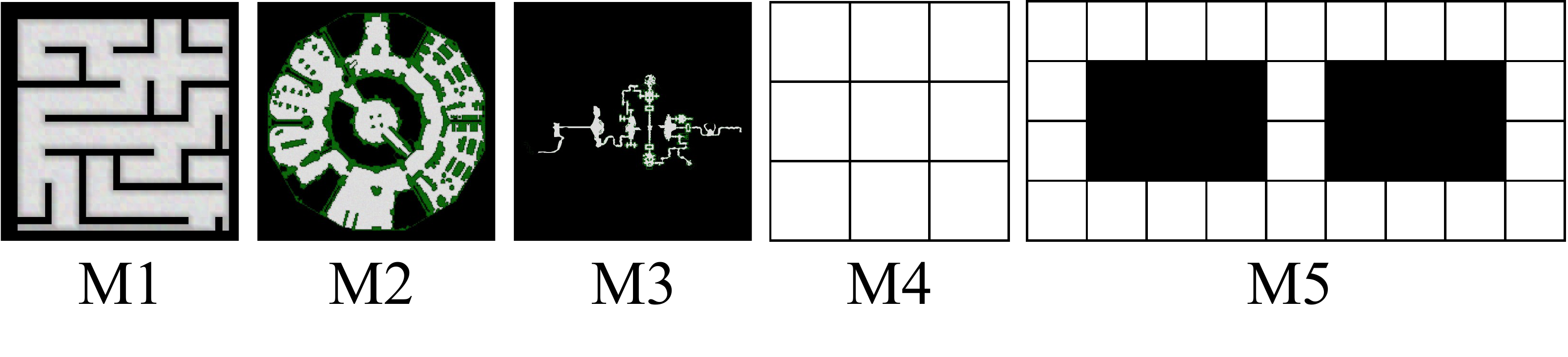} 
\caption{
The grid maps used in the experiment. \label{fig:map} }
\end{figure}
Table \ref{tab:large_succ} and Table \ref{tab:large_run} show the success rate, the average running time, and the speedup ratio to A1 of the other two on M1, M2, and M3. In most instances, A3 outperforms A1 and A2 regarding average running time and success rate. Specifically, the speedup ratio of A3 to A1 is nearly 7 when the number of agents is two on M1. When the number of agents is smaller, the acceleration is more pronounced because when the number of agents is large, unsuccessful instances will smooth the speedup ratio. On M3, A1 and A2 cannot solve any instances, while A3 solves all instances successfully when the number of agents is 2. In small numbers of instances on M2 with few agents, the speedup by using TIS Forest to find the shortest path cannot overlap the construction overhead, making the success rate lower than A1 and A2. However, regarding the average running time, A3 performs the best on the same map and the number of agents. In some instances, A2 runs slightly slower than A1 because A2, which searches the GSI visiting order in the low-level solver, has a higher computation complexity than A1, which only searches the goal vertex visiting order.

\subsection{Experiment for Optimality}
We use two self-defined grid maps, M4 and M5, in Figure \ref{fig:map}, to evaluate the optimality of our proposed algorithm. We generate 100 instances for the number of agents and goals ranging from 2 to 4, with a total of 900 instances for each grid map. We utilize relative error to evaluate the solution quality with A4 and the target algorithm (A1 and A3), only considering instances where both A4 and the target algorithm are successful. Specifically, we count the instance number where their costs differ and compute the maximum and average values of the relative error.

Table \ref{tab:opt} displays the number of instances in which the cost differs between the target algorithms and A4 across all successfully solved instances, as well as the maximum and average relative error. In some instances, A1 cannot obtain the optimal path with a maximal relative error exceeding 17\%, while A3 can obtain optimal results among all instances.

\section{Conclusion}
This work used the case study and experiment to demonstrate that the method based on decoupling the goal vertex visiting order search and the single-agent pathfinding is not optimal for the multi-goal multi-agent pathfinding problem. Hence, we proposed a two-level optimal and efficient solver, MGCBS, decoupling the goal safe interval visiting order search and the single-agent pathfinding. To obtain the shortest path and its length to each goal safe interval efficiently, we proposed the Time-Interval-Space Forest to maintain the shortest path from any start vertex at any start time step to the goal safe interval. Experiments have shown that MGCBS can consistently obtain the optimal result and significantly outperform the SOTA decoupled method regarding running speed.

\section*{Acknowledgments}
This work was supported by Guangdong Basic and Applied Basic Research Foundation (No. 2021B1515120032), and Guangzhou-HKUST(GZ) Joint Funding Program (No. 2024A03J0618). 

\bibliographystyle{named}
\bibliography{ijcai24}
\appendix



\begin{table}[tb]

    \small
    \centering
    \begin{tabular}{ccrrrr}
        \toprule
        $k$ & $n$ &  A1&A2&A3 \\  
        \midrule
2&4&\textbf{100\%} &\textbf{100\%} & \textbf{100\%}  \\ 
&8&\textbf{100\%} &\textbf{100\%} &99\%  \\ 
&12&\textbf{100\%} &\textbf{100\%} &99\%  \\
&16&82\% &82\% &\textbf{99\%}  \\ 
\hline
4&4&98\% &98\% &\textbf{99}\%  \\ 
&8&\textbf{98\%} &\textbf{98\%} &96\%  \\ 
&12&\textbf{91\%} &\textbf{91\%} &89\%  \\ 
&16&46\% &45\% &\textbf{86\%}  \\ 
\hline
6&4&\textbf{94\%} &\textbf{94\%} &82\%  \\ 
&8&\textbf{85\%} &\textbf{84\%} &74\%  \\ 
&12&70\% &69\% &\textbf{80\%}  \\ 
&16&12\% &12\% &\textbf{63\%}  \\ 
\hline
8&4&\textbf{83\%} &\textbf{83\%} &65\%  \\ 
&8&\textbf{66\%} &\textbf{66\%} &50\%  \\ 
&12&18\% &17\% &\textbf{44\%}  \\ 
&16&0\% &0\% &\textbf{40\%}  \\ 
        \bottomrule
    \end{tabular}
    \caption{Table of the success rate. $k$ is the number of agents. $n$ is the number of goals.}
    \label{tab:succ}
\end{table}

\begin{table}[tb]
\small
    \centering
    \begin{tabular}{ccrrrr}
        \toprule
       $k$ & $n$ & A1&A2&A3 \\  
        \midrule

2&4&\textbf{1.67(-)}&\textbf{1.67(1.00)}&1.87(0.89) \\ 
&8&4.06(-)&4.04(1.00)&\textbf{2.76(1.47)} \\ 
&12&9.28(-)&9.24(1.00)&\textbf{3.25(2.86)} \\
&16&28.85(-)&28.79(1.00)&\textbf{3.88(7.44)} \\ 
\hline
4&4&4.84(-)&4.84(1.00)&\textbf{4.44(1.09)} \\ 
&8&10.23(-)&10.22(1.00)&\textbf{7.73(1.32)} \\ 
&12&25.04(-)&25.01(1.00)&\textbf{12.89(1.94)} \\ 
&16&48.48(-)&48.57(1.00)&\textbf{17.23(2.81)} \\ 
\hline
6&4&\textbf{11.26(-)}&\textbf{11.26(1.00)}&18.78(0.60) \\ 
&8&24.32(-)&24.48(0.99)&\textbf{24.27(1.00)} \\ 
&12&41.09(-)&41.12(1.00)&\textbf{23.83(1.72)} \\ 
&16&58.11(-)&58.14(1.00)&\textbf{32.31(1.80)} \\ 
\hline
8&4&\textbf{20.1(-)}&20.13(1.00)&30.84(0.65) \\ 
&8&\textbf{38.17(-)}&38.35(1.00)&41.26(0.93) \\ 
&12&56.79(-)&56.84(1.00)&\textbf{44.64(1.27)} \\ 
&16&60.0(-)&60.0(1.00)&\textbf{47.35(1.27)} \\ 
        \bottomrule
    \end{tabular}
    \caption{Table of the average running time and the speedup ratio. The first number in the cell shows the average running time (sec), and the number in parentheses indicates the speedup ratio relative to A1. $k$ is the number of agents. $n$ is number of goals.}
    \label{tab:run}
\end{table}

\section{Experiments on Various Numbers of Goals}
In this section, we verify the efficiency of the MGCBS on various numbers of goals.
We use the grid map `lak303d', which is M2 in the main text for the experiment.
We generate 100 instances for the number of goals in the range of \{4, 8, 12, 16\} and the number of agents in the range of \{2, 4, 6, 8\}. We randomly generate a start point and goal points for each agent, while the start point for each agent is distinct. The success rate and the average running time are used as evaluation metrics. An instance is considered unsuccessful if the running time exceeds 60 seconds and the running time of the unsuccessful instance is directly set to 60 seconds.

We use three algorithms for our evaluation. The markings are consistent with those in the main text.
\begin{itemize}
\item \textbf{HCBS (A1)}: A three-level algorithm \cite{surynek2021multi}. The highest level solver is a CBS algorithm, and the middle level solver is an A* algorithm that searches a goal vertex visiting order. The lowest level solver is a classic A* for single-agent pathfinding. It is the current SOTA method based on DVS. 
\item \textbf{MGCBS without TIS Forest (A2)}: MGCBS with a modified low-level solver that uses the A* algorithm to obtain the shortest path and its length to a GSI. 
\item \textbf{MGCBS (A3)}: Our propose method. 
\end{itemize}

Table \ref{tab:succ} shows the success rate, and Table \ref{tab:run} shows the average running time of three algorithms. When the number of agents is fixed, A3 outperforms A1 and A2 when the number of goals is large in terms of success rate and average running time. The speedup ratio of A3 to the other two algorithms increases with the increase of the number of agents. On the one hand, when the number of goals is small, the middle-level solver of A1 and the low-level solver of A2 can obtain the goal vertex visiting order or the GSI visiting order with few trials in the search. In these cases, the single-agent pathfinding solver is not called for many times, and the A3, which uses the TIS Forest to reduce the redundant computation of the single-agent pathfinding, is not beneficial. The construction time of the TIS Forest causes a decrease in efficiency. On the other hand, when the number of goals is large, it takes serval trials to find the best goal vertex visiting order and the GSI visiting order, making many calls to the single-agent pathfinding solver in A1 and A2. The TIS Forest in A3 can significantly reduce redundant calculations and improve computational speed.

The average running time of A2 is almost equal to A1 in the experiments, and when the number of agents and the number of goals are large, A2 runs a bit slower than A1. This is because when the number of agents and goals are small, the number of conflicts is small and there is almost no more than one GSI at each goal vertex. When the number of agents and goals are large, there might exist several GSI at the goal vertices, and the search of the GSI visiting order is slower than the search of the goal visiting order.

\end{document}